\documentclass[11pt]{article}
\usepackage{xcolor}
\usepackage[colorlinks=true,citecolor=black,urlcolor=gray, filecolor=black, linkcolor=black]{hyperref}
\usepackage{fullpage}
\usepackage{amsfonts, amsmath, amssymb}

\DeclareMathOperator{\polylog}{polylog}
\definecolor{lightgray}{rgb}{.85,.85,.85}

\newcommand {\bbR}    {\mathbb{R}}

\newcommand {\calP}   {{\cal{P}}}

\newtheorem{theorem}{Theorem}

\newtheorem{claim}[theorem]{Claim}

\newcommand{\qed}{\rule{7pt}{7pt}}

\newenvironment{proof}{\noindent{\bf Proof}\hspace*{1em}}{\qed\bigskip}
\newenvironment{proof-sketch}{\noindent{\bf Sketch of Proof}\hspace*
{1em}}{\qed\bigskip}
\newenvironment{proof-idea}{\noindent{\bf Proof Idea}\hspace*{1em}}
{\qed\bigskip}
\newenvironment{proof-of-lemma}[1]{\noindent{\bf Proof of Lemma #1}
\hspace*{1em}}{\qed\bigskip}
\newenvironment{proof-attempt}{\noindent{\bf Proof Attempt}\hspace*
{1em}}{\qed\bigskip}

\newenvironment{remark}{\noindent{\bf Remark}\hspace*{1em}}{\bigskip}

\makeatletter
\def\fnum@figure{{\bf Figure \thefigure}}
\def\fnum@table{{\bf Table \thetable}}
\long\def\@mycaption#1[#2]#3{\addcontentsline{\csname
   ext@#1\endcsname}{#1}{\protect\numberline{\csname
   the#1\endcsname}{\ignorespaces #2}}\par
   \begingroup
     \@parboxrestore
     \small
     \@makecaption{\csname fnum@#1\endcsname}{\ignorespaces #3}\par
   \endgroup}
\def\mycaption{\refstepcounter\@captype \@dblarg{\@mycaption\@captype}}
\makeatother

\newcommand{\mathify}[1]{\ifmmode{#1}\else\mbox{$#1$}\fi}
\newcommand{\E}{\mathbb{E}}
\newcommand{\bigO}O

\newcommand{\eps}{\varepsilon}

\newcommand{\calV}{\mathcal{V}}

\newcommand{\R}{{\mathbb R}}

\newcommand{\ignore}[1]{}

\newcommand{\email}[1]{\href{mailto:#1}{{\tt #1}}}

\newtheorem{tempclaim}{Claim}

\newcommand{\TmpClaim}[2]{\setcounter{tempclaim}{#1}\addtocounter{tempclaim}{0}\begin{tempclaim}#2\end{tempclaim}}
    
\title{Improved Approximation for the Directed Spanner Problem}
\author{Arnab Bhattacharyya\footnote{Supported in part by NSF Awards
    0514771, 0728645, and 0732334. Part of the work was completed
    during a visit to the IBM T.J.~Watson Research Center.}
\\CSAIL, MIT\\\email{abhatt@mit.edu} 
\and Konstantin Makarychev\\IBM T.J.~Watson Research Center\\\email{konstantin@us.ibm.com}}

\date{}

\begin{document}
\maketitle

\begin{abstract}
We prove that the size of the sparsest directed $k$-spanner of a graph
can be approximated in polynomial time to within a factor of
$\tilde{O}(\sqrt{n})$, for all $k \geq 3$.  This improves the
  $\tilde{O}(n^{2/3})$-approximation recently shown by Dinitz   
and Krauthgamer~\cite{DK10}.  
\end{abstract}

\section{Introduction}
A spanner of a graph generally denotes a sparse subgraph which
preserves all pairwise distances up to a given approximation.  More
specifically, given a graph $G = (V,E)$ and an integer $k \geq 1$,
define the {\em $k$-spanner} of $G = (V,E)$ to be a subgraph $H = (V,
E_H)$ such that, for any two vertices $u, v \in V$,
\begin{equation*}
dist_H(u,v) \leq k \cdot dist_G(u,v)
\end{equation*}
If each edge of graph $G$ has an associated nonnegative length, then
$dist_G(u,v)$ denotes the smallest sum of the lengths of edges along a path
from $u$ to $v$. Spanners have numerous applications, such as efficient
routing~\cite{Cowen01, CowenWagner, PelUp89, roundtrip-spanners,
  TZ01b}, simulating synchronized protocols in unsynchronized 
networks~\cite{PelegU89}, parallel, distributed and streaming
algorithms for approximating shortest paths~\cite{coh98, coh00, e01, FKMSZ08}, and algorithms
for distance oracles~\cite{BasSen, approx-distance-oracles}. 

For integer $k \geq 1$, the computational problem {\sc Directed
  $k$-Spanner} is the task of finding the minimum number of edges in
a $k$-spanner of an input directed graph $G$.    Peleg and
Sch\"affer~\cite{PelegSchaffer} and   Cai~\cite{Cai} show that
 {\sc Directed  $k$-Spanner} is \textsf{NP}-hard for every $k
\geq 2$.  The approximability of {\sc Directed $k$-Spanner} has also
been well-studied.  When $k=2$, Kortsarz and Peleg~\cite{KortPel}, and
Elkin and Peleg~\cite{ElkinPeleg01}, showed a tight~$O(\log n)$
approximation. For $k=3$ and $k=4$, Berman, Raskhodnikova and Ruan
\cite{BRR10} designed an $\tilde{O}(\sqrt{n})$-approximation
algorithm.  Dinitz and Krauthgamer~\cite{DK10} independently showed
$\tilde{O}(\sqrt{n})$ approximability for $k=3$ and gave an
$\tilde{O}(n^{2/3})$-approximation algorithm valid for all $k \geq 4$,
thus obtaining for the first time for this problem an approximation 
ratio that does not degrade with increasing $k$.  It is known
\cite{ElkinPeleg07} that for every $\eps, \delta \in (0,1)$,
{\sc Directed $k$-Spanner} with $3 \leq k = o(n^{1-\delta})$ is
inapproximable in polynomial time to within a factor of
$2^{\log^{1-\eps}n}$, unless $\mathsf{NP} \subseteq
\mathsf{DTIME}(n^{\polylog n})$.   

Our main result is the following.
\begin{theorem}
For any integer $k \geq 3$, there is a polynomial time algorithm that,
in expectation, finds an $\tilde{O}(\sqrt{n})$-approximation
for the {\sc Directed $k$-Spanner} problem.
\end{theorem}
This matches the $\tilde{O}\sqrt{n})$ bound conjectured by Dinitz and
Krauthgamer.  Our algorithm, similar to that of \cite{DK10} and the
earlier \cite{tc-spanners-soda}, operates by letting the edges of the
directed $k$-spanner be the union of two subsets of edges of the
original graph, the first obtained by rounding the solution to
a linear programming (LP) relaxation of the problem and the second obtained
by edges of shortest-path trees growing from randomly selected
vertices.  Our main technical innovation is in designing and analyzing
a new rounding algorithm that selects each edge of the original graph
independently with probability proportional to its LP value.  

\section{Proof of Main Result}

\subsection{The LP Relaxation}
The LP relaxation that we use for the {\sc Directed
  $k$-Spanner} problem is exactly the same as the flow-based LP introduced
by Dinitz and Krauthgamer \cite{DK10}.  Let us reproduce it here.
Suppose we are given a directed graph $G = (V,E)$ with a length
function $\ell: E \to \R^{\geq 0}$.  For every edge $(u,v) \in E$, define
$\mathcal{P}_{u,v}$ to be the set of all directed paths $p$ in $G$ from $u$ to
$v$ such that the length of $p$ is at most $k$ 
times the length of the shortest path in $G$ from $u$ to $v$.  Thus,
for every edge $(u,v) \in E$, the $k$-spanner of $G$ must contain at
least one path from $\mathcal{P}_{u,v}$.  The LP will have
variables $x_e$ for each edge 
$e \in E$, representing whether or not $e$ is included in the spanner,
and variables $f_p$ for each path $P$ in $\bigcup_{(u,v) \in E}
\mathcal{P}_{u,v}$, representing the flow along the path $P$. The LP
is then as follows:
\begin{center}
\colorbox{lightgray}{
\begin{minipage}{3in}
$$\min \sum_{e \in E} x_e$$
subject to:
\begin{align*}
&\sum_{P \in \mathcal{P}_{u,v}: e \in P}f_p \leq x_e & \forall (u,v)
\in E, \forall e \in E\\
&\sum_{P \in \mathcal{P}_{u,v}} f_p \geq 1 & \forall (u,v) \in E\\
&x_e \geq 0 & \forall e \in E\\
&f_p \geq 0 & \forall (u,v) \in E, \forall P \in \mathcal{P}_{u,v}
\end{align*}
\end{minipage}
}
\end{center}

Notice that the number of variables can be exponential for
large $k$, and hence, \textit{a priori}, it might not be clear that
the LP can be solved optimally in polynomial time.  However, by a %%standard
separation oracle argument, one can find an approximate solution in polynomial
time for any $k$.
\begin{theorem}[Theorem 2.1 in Dinitz and Krauthgamer~\cite{DK10}]
There is a polynomial time $(1+\eps)$-approximation algorithm for the
above LP for any constant $\eps > 0$.
\end{theorem}

\subsection{Constructing the \texorpdfstring{$k$}{k}-spanner}
Given a fractional solution $(x,f)$ to the above LP, we now describe how
to construct a $k$-spanner $H = (V,E_H)$ for the input graph $G =
(V,E)$.  The edge set $E_H \subseteq E$ is constructed via the following simple randomized algorithm (with a
parameter $\alpha$):
\begin{center}
\colorbox{lightgray}{
\begin{minipage}{6in}
\begin{enumerate}
\item
Let $E_H = \varnothing$.
\item
For each edge $e \in E$ independently, add $e$ to $E_H$ with
probability $\min(\alpha x_e \sqrt{n}, 1)$. 
\item
Let $S$ be a set of vertices, each chosen independently at random from $V$ with probability
$\min (\alpha /\sqrt{n}, 1)$.  For each $v \in S$, add
the edges of the outward shortest-path tree and the inward shortest-path tree rooted at $v$ to $E_H$.
\end{enumerate}
\end{minipage}}
\end{center}

We first give a proof for a special case when all edges have unit
lengths (i.e., $\ell(e) = 1$ for all $e \in E$) and $k$ is upper-bounded by a
constant.  Then, we give an alternative proof that works for general
edge lengths and arbitrary $k$. 

\subsection{Unit length edges}

\begin{theorem}\label{thm:t1}
For graphs with unit edge lengths, the approximation ratio of the
randomized algorithm described above 
(with $\alpha = 10\sqrt{k} \log n$) is $30\sqrt{nk}\log n$.
\end{theorem}
\begin{proof}
We first show that the expected cost of the solution is at most $30\sqrt{nk}\log n$ times the optimal cost. Then, we prove that the obtained solution is feasible.

The expected number of edges sampled at the second step is $10\sqrt{nk} \log n\; LP$, where $LP$ is the cost of the LP. The cost of every shortest-path tree is at most the cost of the optimal solution, since the 
optimal solution must contain at least $(n-1)$ edges (here we assume that $G$ is connected, since we can handle different connected components separately). Thus, the expected number of edges added at the third
step is at most $20\sqrt{nk} \log n~OPT$. The total expected cost of the solution is at most 
$30\sqrt{nk} \log n\;  OPT$.

We now prove that the solution returned by the algorithm is feasible with probability $1-n^{-3}$. It suffices
to show that with probability at least $1-n^{-3}$, for every adjacent vertices $u$ and $v$ ($(u,v)\in E$)
there exists a path in $H$ of length at most $k$ connecting $u$ and $v$. Fix an arbitrary $(u,v)\in E$.

Let as before $\mathcal{P}_{u,v}$ be the set
of all paths of length at most $k$ going from $u$ to $v$, and let
$$\mathcal{V}_{u,v} = \bigcup_{p\in \mathcal{P}_{u,v}} p$$
be the set of vertices covered by these paths. 
We consider two cases. If $|\mathcal{V}_{u,v}| \geq \sqrt{n/k}$, then the set $S$ sampled at the third
step of the algorithm contains at least one vertex from $\mathcal{V}_{u,v}$
with probability 
$$1 - \Big(1- \frac{10\sqrt{k}\log n}{\sqrt{n}}\Big)^{|\mathcal{V}_{u,v}|}\geq 
1 - \Big(1- \frac{10\log n}{\sqrt{n/k}}\Big)^{\sqrt{n/k}}\geq 1 - e^{-10 \log n} = 1-n^{-10}.$$
If $ S\cap \mathcal{V}_{u,v} \neq \varnothing$, we pick an $s\in S\cap \mathcal{V}_{u,v}$. Since, $s\in  \mathcal{V}_{u,v}$, $dist_G(u, s) + dist_G(s,v)\leq k$.
Hence, the union of the inward and outward shortest-path trees rooted at $s$ (that the algorithm adds to $E_H$ at the third step) contains a path going from $u$ to $v$ of length at most $k$.

We now consider the case $|\mathcal{V}_{u,v}| \leq \sqrt{n/k}$. Perform a mental experiment. Make $k$
copies of each vertex $w\in \mathcal{V}_{u,v}\setminus\{u\}$: $(w,1),(w,2),\dots, (w,k)$; and make a copy
of vertex $u$: $(u,0)$. For every
path $p\in \mathcal{P}_{u,v}$ define a new path $\tilde{p}$ from $u$ to $v$ as
$$\tilde{p}_i =  (p_i,i)$$
where $\tilde{p}_i$ and $p_i$ are the $i$-th vertices of $\tilde{p}$ and $p$ respectively. 
We assign flow $f_p$ to $\tilde{p}$. The total flow going from $(u,0)$ to $(v,k')$ (for $k'\leq k$)
in the new graph is at least $1$. We now start removing ``light'' vertices from the new graph. Initially, let $\widetilde{\calV} = \calV_{u,v}\times\{1,2,\dots, k\}\cup \{(u,0)\}$ (i.e., $\widetilde{\calV}$ is the set of all ``copies'' $(w,i)$) and $\widetilde{\calP} = \{\tilde{p}:p\in \calP_{u,v}\}$.
Note, that $|\widetilde{\calV}| \leq |\calV| \times k \leq \sqrt{nk}$.  At every 
step we find a vertex $(w,i)\in \widetilde{\calV}$, $w\ne v$ the flow through which is less than $1/\sqrt{2nk}$ (we call
this vertex a ``light'' vertex) i.e., a vertex $(w,i)\in \widetilde{\calV}$, $w\ne v$ such that
$$\sum_{\tilde{p}\in \widetilde{\calP}: w\in {p}} f_p \leq \frac{1}{\sqrt{2nk}}.$$
We then remove $(w,i)$ from $\widetilde{\calV}$ and all paths $\tilde{p}\ni (w,i)$ from $\widetilde{\calP}$.
We stop when there are no more such vertices $(w,i)$ left. Note, that after removing each vertex $w$ and all
paths going through it, we recompute the flow going through the remaining vertices. 

We remove at most $\sqrt{nk}$ vertices from $\widetilde{\calV}$ (simply because before removing ``light'' vertices the size of $\widetilde{\calV}$ was at most $\sqrt{nk}$), and at most $\sqrt{nk}/\sqrt{2nk} = \sqrt{2}/{2}$ units of flow. Thus, the remaining
weight of paths $\tilde{p}\in \widetilde{\calP}$ is at least $1-\sqrt{2}/{2} > 1/4$; and
thus  $(u,0)\in \widetilde{\calV}$.

Now, pick an arbitrary vertex $(w,i)\in \widetilde{\calV}$. Observe, that if $w\neq v$,
$$\sum_{(w,w')\in E\colon (w',i+1) \in \widetilde{\calV}} x_{(w,w')} \geq
\sum_{p\colon p_i=w,\; p_{i+1}\in \widetilde{\calV},\; p\in \calP_{u,v}} f_{p} 
\geq 
\sum_{\tilde{p}\in \widetilde{\calP}: w \in p} f_{p} \geq \frac{1}{\sqrt{2nk}}. 
$$
Hence, the set $E_H$ contains at least one edge $(w,w')$ such that $(w',i+1) \in \widetilde{V}$
with probability
\begin{eqnarray*}
1 - \prod_{(w,w')\colon (w',i+1) \in \widetilde{\calV}} (1 - 10 x_{(w,w')} \sqrt{nk}\log n) &\geq&
1 - \prod_{(w,w')\colon (w',i+1) \in \widetilde{\calV}} e^{- 10 x_{(w,w')} \sqrt{nk}\log n}\\
 &\geq& 1 - e^{-5\sqrt{2} \log n}> 1 - n^{-7},
\end{eqnarray*}
if for all  $(w',i+1)\in \widetilde{\calV}$, $x_{(w,w')} \leq (10 \sqrt{nk}\log n)^{-1}$; and
with probability $1$, otherwise. By the union bound,
with probability at least $1 - n^{-5}$ for all $(w,i)\in \widetilde{V}$ there is such $(w,w')$ in $E_H$. Thus,
there exists a path $w_1=u, w_2, \dots, w_{k'}=v$ of length at most $k$ such that $(w_i,i)\in \widetilde{V}$
and $(w_i, w_{i+1})\in E_H$ for every $i$. Hence, $u$ and $v$ are connected with a path of length at most $k$
with probability at least $1 - n^{-5}$.

We showed that for every $(u,v)\in E$, with probability $1-n^{-5}$, there exists a path in $H=(V,E_H)$
of length at most $k$ connecting $u$ and $v$; and therefore, with probability at least $1-1/n^{-3}$, $E_H$ is 
a $k$-spanner.
\end{proof}

\subsection{General case}

\begin{theorem}\label{thm:t2}
The approximation ratio of the randomized algorithm described above (with $\alpha = 5\log n$) is $15\sqrt{n}\log n$.
\end{theorem}

\begin{remark}
The first several steps in the proof are the same as in the proof of Theorem~\ref{thm:t1}. However,
the proof of the main case ($|\mathcal{V}_{u,v}|\leq \sqrt{n}$) is very different from 
the previous proof.
\end{remark}

\begin{proof}
The expected number of edges sampled at the second step is $5\sqrt{n}
\log n\; LP$, where $LP$ is the cost of the LP. The cost of every
shortest-path tree is at most the cost of the optimal solution, since
the  optimal solution must contain at least $(n-1)$ edges (here we assume
that $G$ is connected, since we can handle different connected
components separately).  In expectation, we add the edges of at
most $2 \alpha \sqrt{n} = 10 \sqrt{n} \log n$ shortest-path trees, and
so, the expected number of edges added at the third step is at most
$10\sqrt{n} \log n~OPT$. The total expected cost of the solution is
at most   $15\sqrt{n} \log n\;  OPT$.

We now prove that the solution returned by the algorithm is feasible with probability at least 
$1-n^{-3}$.  Consider two arbitrary adjacent vertices $u$ and $v$ ($(u,v)\in E$). We show that with probability $1-n^{-5}$, there exists a path in $H$ of length at most $k\ell (u,v)$ 
connecting $u$ and $v$.

As in Theorem~\ref{thm:t1} we consider two cases:
the set $\mathcal{V}_{u,v} = \bigcup_{p\in \mathcal{P}_{u,v}} p$ is large 
($|\mathcal{V}_{u,v}|\geq \sqrt{n}|$) and small ($|\mathcal{V}_{u,v}|\leq \sqrt{n}|$).
For the case $|\mathcal{V}_{u,v}|\geq \sqrt{n}$, we use exactly the same proof as before to show
that $u$ and $v$ are connected with a short path which is contained in the union of two
shortest-path trees (see Theorem~\ref{thm:t1}).

So, consider the second case, $|\mathcal{V}_{u,v}|\leq \sqrt{n}$. 
Let $G'= G[\calV_{u,v}]$ be the graph induced on the vertex set $\calV_{u,v}$.
We denote the set of edges of $G'$ by $E' = \{(w_1,w_2)\in G:w_1,w_2\in \calV_{u,v}\}$, $H'=H\cap G'$ and 
$E_H' = E_H \cap E'$. We show that
the set $E_H$ satisfies the dual LP constraints with high probability.

Consider an arborescence $T\subset G'$ rooted at $u$. Let $d_T$ be the shortest 
path metric on $T$. Define a function $L_T: \calV_{u,v}\to \bbR^+$,  
$$L_T(w) = d_T(u,w),$$ 
and let 
$$S_T = \{(w_1,w_2)\in E' : L_T(w_2) > L_T(w_1) + \ell(w_1,w_2)\}.$$

In Section~\ref{sec:claims}, we prove the following claims.

\begin{claim}\label{claim:c1}
A subgraph $H'\subset G'$ contains a directed path of length at most $K$ connecting $u$ and $v$, if and only if for every arborescence $T\subset G'$ rooted at $u$ with $d_T(u,v)> K$, 
$$
E_H' \cap S_T \neq \varnothing.
$$
\end{claim}

\begin{remark}
This claim is an analog of the special case of min-cut/max-flow theorem: A subgraph $H'$ contains 
a directed path connecting $u$ and $v$, if and only if every cut separating $u$ and $v$ in $G'$ contains an edge from $E_H'$.
\end{remark}

\begin{claim}\label{claim:c2}
For every arborescence $T\subset G'$ rooted at $u$ with $d_T(u,v)>
k\cdot \ell(u,v)$, 
$$\sum_{e\in S_T} x_e \geq 1,$$
where $\{x_{e}\}$ is the LP solution.
\end{claim}

We show that for every arborescence $T\subset G'$ rooted at $u$ with $d_T(u,v) > k \cdot \ell(u,v)$, 
$$E_H' \cap S_T \neq \varnothing$$
with high probability;
and, thus, by Claim~\ref{claim:c1} there exists a path in $E_H$ 
of length at most $k \cdot \ell(u,v)$ connecting $u$ and $v$.
Let $X_e$ be the indicator random variable defined as follows: $X_e = 1$ if $e\in E_H$; $X_e = 0$ otherwise. 
If for some $e\in S_T$, $x_e \geq 1/(\alpha\sqrt{n})$, then $X_e = 1$ with probability 1. Otherwise, for every $e\in S_T$, 
$\Pr (X_e = 1) = \alpha \sqrt{n} x_e$.

\medskip

Estimate the probability that for all $e\in S_T$, $X_i=0$,
\begin{eqnarray*}
\Pr(X_e = 0 \text { for all } e\in S_T) 
&=& \prod_{e\in S_T}(1 - \Pr(X_e = 1)) \leq \prod_{e\in S_T} e^{-\Pr(X_e = 1)}\\
&=& e^{-\sum_{e\in S_T}\Pr(X_e = 1)} =  e^{-\sum_{e \in S_T}  \alpha  x_e\sqrt{n}} \\
&=& e^{- \alpha  \sqrt{n} \Big(\sum_{e \in S_T} x_e\Big)}.
\end{eqnarray*}
By Claim~\ref{claim:c2},
$$\sum_{e\in S_T} x_e \geq 1,$$
thus,
$\Pr(X_e = 0 \text { for all } e\in S_T) \leq e^{-\alpha \sqrt{n}}$.
Hence, 
$\Pr(\cup_{e\in S_T} \{X_e = 1\}) \geq 1 - e^{-\alpha \sqrt{n}}$. In other words, for a fixed arborescence 
$T$, 
$$\Pr(S_T\cap E'_H) \geq 1 - e^{-\alpha \sqrt{n}}.$$
The total number of arborescences can be bounded by
$$|\calV_{u,v}|^{|\calV_{u,v}|}\leq \sqrt{n}^{\sqrt{n}} = e^{1/2\;\sqrt{n}\log n},$$
since for every vertex $w\in \calV_{u,v}$ there at most $|\calV_{u,v}|$ possible ways to choose a parent
node $w'\in \calV_{u,v}$. Hence, by the union bound with probability at least 
$$1 - e^{-\alpha \sqrt{n}}e^{1/2\;\sqrt{n}\log n} \geq  1 - e^{-\sqrt{n}\log n}$$
there exists a path of length at most $k \cdot \ell(u,v)$ between $u$ and $v$ in $H$.
\end{proof}

\subsection{Proofs of Claim~\ref{claim:c1} and Claim~\ref{claim:c2}}\label{sec:claims}

\TmpClaim{4}{%%
A subgraph $H'\subset G'$ contains a directed path of length at most $K$ connecting $u$ and $v$, if and only if for every arborescence $T\subset G'$ rooted at $u$ with $d_T(u,v)> K$, 
$$
E_H' \cap S_T \neq \varnothing.
$$
}

\begin{proof}
\begin{enumerate}
\item[I.] Suppose that there is a path $p$ of length at most $K$ connecting $u$ and $v$ in $H'$. Consider an arbitrary arborescence $T\subset G'$ rooted at $u$ with $d_T(u,v)> K$. Write, 
$$\text{length(p)}\equiv \sum_{i=1}^{|p|-1} \ell(p_i, p_{i+1}) \leq K.$$
Then,
$$\sum_{i=1}^{|p|-1} (L_T(p_{i+1}) - L_T(p_{i})) = L_T(p_{|p|}) - L_T(p_1) = L_T(v) - L_T(u) > K.$$
Hence, for some $i$,
$$\ell(p_i, p_{i+1}) < L_T(p_{i+1}) - L_T(p_{i}),$$
and therefore, $(p_{i}, p_{i+1})\in S_T$ (by the definition of $S_T$). Since, $(p_{i}, p_{i+1})\in E'_H$,
$S_T \cap E'_H \neq \varnothing$.

\item[II.] Now, assume that for every arborescence $T\subset G'$ rooted at $u$ with $d_T(u,v)> K$, 
$$
E_H' \cap S_T \neq \varnothing.
$$

Let $T$ be the directed shortest path tree in $E_H'$ rooted at $u$. We claim that 
$S_T\cap E'_H=\varnothing$ and thus $d_T(u,v) \leq K$. Indeed, for all $(w_1,w_2)\in E_H'$,
$$L_T(w_2)\equiv d_T(u,w_2) \leq d_T(u,w_1)+\ell(w_1,w_2) \equiv L_T(w_1) + \ell(w_1,w_2),$$ 
here we have used that $T$ is the shortest path tree and thus $d_T(u,w_2) \leq d_T(u,w_1)+\ell(w_1,w_2)$.
Therefore, $(w_1,w_2)\notin S_T$.
\end{enumerate}
\end{proof}

\TmpClaim{5}
{
For every arborescence $T\subset G'$ rooted at $u$ with $d_T(u,v)> k
\cdot \ell(u,v)$,
$$\sum_{e\in S_T} x_e \geq 1,$$
where $\{x_{e}\}$ is the LP solution.
}
\begin{proof}
Since $\{x_e\}$ is part of an LP solution:
\begin{align*}
\sum_{e \in S_T} x_e 
&\geq \sum_{e \in S_T} \sum_{p \in \calP_{u,v} : e \in p} f_p\\
&\geq \sum_{p \in \calP_{u,v}} f_p \\
&\geq 1
\end{align*}

The first and third inequalities are from the definition of the LP.
The second inequality follows because, by Claim \ref{claim:c1}, every
path $p \in \calP_{u,v}$ contains at least one edge in $S_T$ and
because each $f_p$ is nonnegative.
\ignore{
Since $\{x_{e}\}$ is an LP solution, there exists a flow $f_p$ routed on paths of length
at most $K=k \cdot \ell(u,v)$ between $u$ and $v$ in graph $G'$ with capacities $\{x_e\}$. Hence, for a given LP solution $\{x_{e}\}$, the value
of the following LP (with variables $f_p$, $p\in \calP_{u,v}$) is at least 1:

\begin{center}
\colorbox{lightgray}{
\begin{minipage}{3in}
$$\max \sum_{p\in \calP_{u,v}} f_p$$
subject to:
\begin{align*}
&\sum_{p \in \mathcal{P}_{u,v}: e \in p}f_p \leq x_e &\forall e \in E\\
&f_p \geq 0 & \forall p \in \mathcal{P}_{u,v}
\end{align*}
\end{minipage}
}
\end{center}

The dual LP (with variables $z_e$, $e\in E'$) is as follows:

\begin{center}
\colorbox{lightgray}{
\begin{minipage}{3in}
$$\min \sum_{e\in \E'} x_e z_e$$
subject to:
\begin{align*}
&\sum_{e \in p}z_e \geq 1 &\forall p \in \calP_{u,v}\\
&z_e \geq 0 & \forall e \in E'
\end{align*}
\end{minipage}
}
\end{center}

Consider a dual LP solution $\{z_e\}$: $z_e = 1$ if $e\in S_T$, and $z_e=0$ otherwise. To show that this
is a feasible solution, consider a path $p\in \calP_{uv}$. By Claim~\ref{claim:c1}, applied to the
set $E'_H$ containing edges of $p$, the path $p$ has a non-empty intersection with $S_T$. Hence,
$\sum_{e \in p} z_e \geq 1$.

By the LP duality theorem, 
$$\sum_{e\in S_T} x_e =  \sum_{e\in E'} x_e z_e \geq 1.$$}
\end{proof}

\section{Conclusion}

We proved above an $\tilde{O}(\sqrt{n})$-approximation for
\textsc{Directed $k$-Spanner}.  This settles the conjecture of 
Dinitz and Krauthgamer~\cite{DK10}. Note that Elkin and Peleg
\cite{ElkinPeleg07} have shown that the approximation ratio for this
problem cannnot be expected to be $O(n^{1/k})$ or even $O(n^{1/\log
  k})$, in contrast to the situation for undirected graphs.  

Our algorithm obviously applies to special cases of the {\sc Directed
  $k$-Spanner} problem, such as the \textsc{$k$-Transitive Closure
  Spanner} problem \cite{tc-spanners-soda}.  It also straightforwardly
extends to the \textsc{Client-Server 
  $k$-Spanner} problem and the \textsc{$k$-Diameter Spanning Subgraph}
problem. See \cite{ElkinPeleg05} for definitions and motivation; we
omit details here. Finally, consider
the \textsc{General Directed $k$-Spanner} problem, defined in
\cite{ElkinPeleg07}.  Here, each edge of the input directed graph has
a nonnegative weight as well as a length, and the objective is to
minimize the sum 
of the weights of the edges in the $k$-spanner.  We note that our
$\tilde{O}(\sqrt{n})$ approximation ratio still applies to this
problem, as long as the weight of each edge is lower-bounded by a
positive constant.

\bibliographystyle{alpha}
\bibliography{spanners-bibliography}

\end{document}